\documentclass[a4paper,onecolumn,10pt,noarxiv]{quantumarticle}\pdfoutput=1
\usepackage[utf8]{inputenc}
\usepackage[english]{babel}
\usepackage[T1]{fontenc}
\usepackage{amsmath}
\usepackage[colorlinks,linkcolor=blue,anchorcolor=blue,citecolor=blue]{hyperref}

\usepackage{subfigure}
\usepackage{lipsum}
\usepackage{bm}

\usepackage{tikz}
\usetikzlibrary{quantikz}
\usetikzlibrary{calc,positioning}

\usepackage{amsthm}

\newtheorem{mythm}{Theorem}

\usepackage{amssymb}

\newcommand{\casql}{Key Laboratory of Quantum Information, Chinese Academy of Sciences, School of Physics, University of Science and Technology of China, Hefei, Anhui, 230026, China}
\newcommand{\casex}{CAS Center For Excellence in Quantum Information and Quantum Physics, University of Science and Technology of China, Hefei, Anhui, 230026, China}
\newcommand{\hfnational}{Hefei National Laboratory, University of Science and Technology of China, Hefei, Anhui, 230088, China}
\newcommand{\aihf}{Institute of Artificial Intelligence, Hefei Comprehensive National Science Center, Hefei, Anhui, 230088, China}
\newcommand{\origin}{Origin Quantum Computing, Hefei, Anhui, 230026, China}

\begin{document}
\title{Can Variational Quantum Algorithms Demonstrate Quantum Advantages? Time Really Matters}

\author{Huan-Yu Liu}
\email{liuhuanyu@ustc.edu.cn}
\orcid{0000-0002-6158-9627}
\affiliation{\casql}
\affiliation{\casex}
\affiliation{\hfnational}

\author{Zhao-Yun Chen}
\orcid{0000-0002-5181-160X}
\email{chenzhaoyun@iai.ustc.edu.cn}
\affiliation{\aihf}

\author{Tai-Ping Sun}
\orcid{0009-0009-2591-1672}
\affiliation{\casql}
\affiliation{\casex}
\affiliation{\hfnational}

\author{Cheng Xue}
\orcid{0000-0003-2207-9998}
\affiliation{\aihf}

\author{Yu-Chun Wu}
\email{wuyuchun@ustc.edu.cn}
\orcid{0000-0002-8997-3030}
\affiliation{\casql}
\affiliation{\casex}
\affiliation{\hfnational}
\affiliation{\aihf}

\author{Guo-Ping Guo}
\orcid{0000-0002-2179-9507}
\affiliation{\casql}
\affiliation{\casex}
\affiliation{\hfnational}
\affiliation{\aihf}
\affiliation{\origin}

\begin{abstract}
    Applying low-depth quantum neural networks (QNNs), variational quantum algorithms (VQAs) are both promising and challenging in the noisy intermediate-scale quantum (NISQ) era: Despite its remarkable progress, criticisms on the efficiency and feasibility issues never stopped. 
    However, whether VQAs can demonstrate quantum advantages is still undetermined till now, which will be investigated in this paper.
    First, we will prove that there exists a dependency between the parameter number and the gradient-evaluation cost when training QNNs. Noticing there is no such direct dependency when training classical neural networks with the backpropagation algorithm, we argue that such a dependency limits the scalability of VQAs. 
    Second, we estimate the time for running VQAs in ideal cases, i.e., without considering realistic limitations like noise and reachability. We will show that the ideal time cost easily reaches the order of a 1-year wall time. 
    Third, by comparing with the time cost using classical simulation of quantum circuits, we will show that VQAs can only outperform the classical simulation case when the time cost reaches the scaling of $10^0$-$10^2$ years.
    Finally, based on the above results, we argue that it would be difficult for VQAs to outperform classical cases in view of time scaling, and therefore, demonstrate quantum advantages, with the current workflow.
    Since VQAs as well as quantum computing are developing rapidly, this work does not aim to deny the potential of VQAs. The analysis in this paper provides directions for optimizing VQAs, and in the long run, seeking more natural hybrid quantum-classical algorithms would be meaningful. 
\end{abstract}
\maketitle

\section{Introduction}\label{intro}

Machine learning (ML) \cite{ml,ml2,ml3} is one of the most remarkable technology in the 21st century, which has applications ranging from daily works to scientific research \cite{mlapp}. Developments of ML rely on the success of computer science and the neural network (NN) model \cite{nn}, which provided the capability of carrying out huge computational tasks and simulating complex functions. Quantum computing \cite{quantumcompute} is also developed rapidly in decades, whose features, like quantum entanglement and quantum operation parallelism, are unavailable for their classical counterparts. Quantum computing has been introduced to the ML region, known as quantum machine learning (QML) \cite{vqml1,vqml2}.

Variational quantum algorithms (VQAs) \cite{vqa,vqe} are representative of QML, whose workflow is shown in Fig. \ref{vqaflow}. It is a hybrid quantum-classical algorithm. A quantum processor prepares an ansatz with the quantum neural network (QNN) \cite{pqcml} $U(\bm{\theta})$ 
\footnote{It is also called parameterized quantum circuits in some works. To make it consistent with classical machine learning, we use QNN here.}
as $| \psi (\bm{\theta}) \rangle= U( {\bm{\theta}} )|0\rangle$ with $\bm{\theta}=\{  \theta_1,\theta_2,\cdots,\theta_L \}$ the (trainable) parameter vector. The ansatz is then used to evaluate cost functions with quantum measurements, which is usually an expectation value under some Hamiltonian $H$: $C(\bm{\theta}) =\langle \psi(\bm{\theta})|H|\psi(\bm{\theta})\rangle$. The classical processor optimizes $\bm{\theta}$ to minimize the cost function. QNNs in VQAs are usually low-depth, which can be performed on current noisy intermediate-scale quantum (NISQ) \cite{nisq} devices even without the support of fault-tolerant quantum computation technology \cite{ftqc}.
This makes VQAs potential to achieve quantum advantages in the NISQ era. Since its proposal, VQAs have been developed rapidly and have 
applications ranging from quantum chemistry simulation \cite{vqachemistry1,vqachemistry2,hea} to numerical computation \cite{vqsvd,vqpe}. Experimental demonstrations have also been performed \cite{hea,vqhf,vqfexp}.

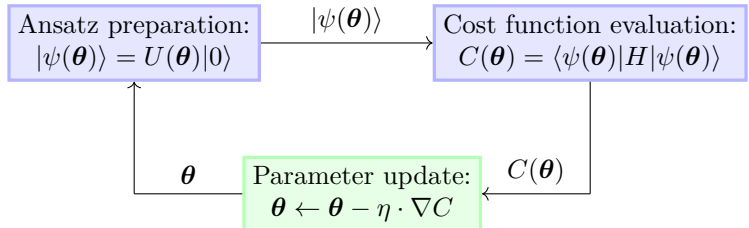
\begin{figure}
    \centering
    \begin{tikzpicture}
    [lnode/.style={rectangle,draw=blue!30,fill=blue!10,align=center,very thick,minimum size=8mm},
    rnode/.style={rectangle,draw=green!30,fill=green!10,align=center,very thick,minimum size=8mm}]
    \node[lnode](qcini) at (-3,0){Ansatz preparation:\\ $|\psi(\bm{\theta})\rangle=U(\bm{\theta})|0\rangle$};
    \node[lnode](qm) at (3,0){Cost function evaluation:\\ $C(\bm{\theta})=\langle \psi(\bm{\theta}) |H|\psi(\bm{\theta})\rangle  $};
    \draw [->] (qcini) -- node[above]{$|\psi(\bm{\theta})\rangle$}  (qm);
    \node[rnode](optimize) at (0,-2){Parameter update:\\ $\bm{\theta}  \leftarrow \bm{\theta} - \eta \cdot \nabla C   $};
    \draw  (optimize) -- node[above]{$\bm{\theta}$}  (-3,-2);
    \draw [->] (-3,-2) -- (qcini);
    \draw (qm) -- (3,-2);
    \draw [->] (3,-2) -- node[above]{$C(\bm{\theta})$} (optimize);
\end{tikzpicture}
    \caption{Workflow of variational quantum algorithms. Operation labeled blue or green indicates that it is performed on quantum or classical processors, respectively.}
    \label{vqaflow}\end{figure}

As research progresses, the challenges of VQAs gradually attracted attention, which can be divided into the efficiency part and feasibility part: Efficiency challenges usually mean that executing VQAs requires huge resources. The well-known barren plateaus \cite{bp} describes a phenomenon with exponentially vanishing gradients, indicating the required sampling times to obtain the cost function also grows exponentially with the number of qubits. On the other hand, feasibility challenges are the major part. They focus on whether the correct answer can be acquired by running VQAs. Training VQAs is an NP-hard problem \cite{vqanphard}, besides the barren plateaus problem mentioned above, there usually exists a variety of local minimum points in the optimization landscape of VQAs \cite{vqalocalmin}, implying that it is difficult to achieve the global optimal point. The expressibility of QNNs \cite{exppqc,exppqc2} also affected the reachability issue \cite{vqereach}, where global optimal points will never be reachable if they cannot be represented by the QNN. Noise \cite{vqanoise,qaoanoise} and other factors will also affect the correctness of executing VQAs. Great efforts have also been provided to deal with such challenges, including mitigating barren plateaus to improve trainability \cite{bpsolve,bpsolve2,bpsolve3}, reducing sampling times to improve efficiency \cite{vqemeasurereducto}, mitigating noises \cite{em,em2}, etc. 

We focus on challenges in the efficiency part in this work. First, we will prove that there exists a dependency between the number of parameters in QNNs and the gradient-evaluation cost when training the QNN. Noticing that such a dependency does not exist when training classical NN models with the backpropagation algorithm \cite{bpnn}, we argue that the parameter number affected the scalability of VQAs. Next, we consider the time cost for running VQAs in an ideal setting, i.e., we do not consider realistic limitations on VQAs like noise, qubit connectivity, reachability, etc. The time cost analysis is used as follows:
\begin{itemize}
    \item The time cost scaling easily reached the 1-year wall time at about 20 qubits.
    \item By comparing with the time cost using classical simulation, we can see that VQAs can only outperform classical simulations when the time cost reaches a scaling of $10^0-10^2$ years. Therefore, quantum advantages are difficult for VQAs to achieve based on the current workflow.
\end{itemize}

In performing such analysis, we would not deny the potential of VQAs, as well as other hybrid quantum-classical algorithms in the NISQ era, but some changes and improvements need to be made. According to our analysis, some directions for optimizing VQAs are provided. Taking one step further, we need to consider what is the natural way of executing machine learning with quantum computing.

The rest of this paper is organized as follows:
In Sec. \ref{Preliminary}, we introduced some backgrounds needed for the latter analysis, including training NNs with the backpropagation algorithm and QNNs.
In Sec. \ref{gradientindepend}, the dependency of the parameter number and the gradient-evaluation cost in training QNNs is provided. 
In Sec. \ref{costvaq}, we analyze the time cost of running VQAs.
Sec. \ref{costresult} gives the total time cost of running VQAs.
In Sec. \ref{qvsc}, we compare the time cost using both VQAs and classical simulation.
A conclusion is given in Sec. \ref{conclusion}.

\section{Preliminary}\label{Preliminary}

\subsection{Training classical neural networks using the backpropagation algorithm}\label{backpropagation}

The NN model is widely applied in solving ML tasks. General NNs are comprised of neurons, whose diagram is shown in Fig. \ref{neuron}. A neuron can be viewed as a non-linear function that maps $n$ inputs $\bm{x}=\{x_1,x_2,\cdots,x_n\}$ to an output $y$ as:
\begin{equation}\label{neuronfunc}
    y = f\left( \sum_i w_ix_i-b  \right),
\end{equation}
where $b$ is a bias,  $\bm{w}=\{ w_1,w_2,\cdots,w_n\}$ is the adjustable weight vector, $f$ is the non-linear activation function and one example is the sigmod function: 
\begin{equation}
    f(x)=\frac{1}{1+e^{-x}}.
\end{equation}
Different functions can be approximated by adjusting the weight vector, and the core idea of ML is to make such functions approach desired maps. ``Learning'' is exactly the process of adjusting the weights.

Only one neuron has limited learning capability. To further increase the expressive power, i.e., be able to fit more functions, neurons can be used to construct a NN, which is shown in Fig. \ref{network}. In the NN, the input is fed into several neurons, whose outputs are then viewed as inputs to neurons in the next layer. Denote $\bm{y}=\{y_1,y_2\cdots,y_m\}$ as the output of the whole NN, or equivalently, the output of neurons corresponding to the final layer. Denote the desired value as $\bm{d}=\{d_1,d_2\cdots,d_m\}$ and the vector of weights for all neurons as $\bm{W}$. As introduced, the learning process is to adjust $\bm{W}$ such that $\bm{y}$ is close to $\bm{d}$.

\begin{figure}[ht]
    \centering
    \subfigure[Neuron]{
        \begin{tikzpicture}
    [lnode/.style={circle,draw=blue!30,fill=blue!10,very thick,minimum size=8mm}]
    \node(x_1) at (-2, -1){$x_1$};
    \node(x_2) at (-2, -2){$x_2$};
    \node(x_3) at (-2, -3){$\vdots$};
    \node(x_4) at (-2, -4){$x_n$};

    \node[lnode](f) at (0,-2.5){$y$};

    \draw [->] (x_1) -- node[above]{$w_1$} (f);
    \draw [->] (x_2) -- node[above]{$w_2$} (f);
    \draw [->] (x_4) -- node[above]{$w_n$} (f);

    \node(y) at (3, -2.5){$y=f\left( \sum_i w_ix_i-b  \right)$};
    \draw [->] (f) --  (y);
\end{tikzpicture}\label{neuron}
    }
    \subfigure[Neural network]{
        \begin{tikzpicture}
    [lnode/.style={circle,draw=blue!30,fill=blue!10,very thick,minimum size=8mm},
    mnode/.style={circle,draw=pink!100,fill=pink!30,very thick,minimum size=8mm},
    rnode/.style={circle,draw=green!30,fill=green!10,very thick,minimum size=8mm}]
    \node(x_1) at (-2, -1){$x_1$};
    \node(x_2) at (-2, -2){$x_2$};
    \node(x_3) at (-2, -3){$\vdots$};
    \node(x_4) at (-2, -4){$x_n$};

    \node[lnode](l_1) at (0,-1){$y_{l1}$};
    \node[lnode](l_2) at (0,-2){$y_{l2}$};
    \node(l_3) at (0,-3){$\vdots$};
    \node[lnode](l_4) at (0,-4){$y_{ln_l}$};

    \node[mnode](m_1) at (2,-0.5){$y_{m1}$};
    \node[mnode](m_2) at (2,-1.5){$y_{m2}$};
    \node[mnode](m_3) at (2,-2.5){$y_{m3}$};
    \node(m_4) at (2,-3.5){$\vdots$};
    \node[mnode](m_5) at (2,-4.5){$y_{mn_m}$};

    \node[rnode](r_1) at (4,-1.5){$y_{1}$};
    \node(r_2) at (4,-2.5){$\vdots$};
    \node[rnode](r_3) at (4,-3.5){$y_{m}$};

    \draw [red,->] (x_1) -- node[above]{$w_0$} (l_1);
    \draw [->] (x_1) --  (l_2);
    \draw [->] (x_1) --  (l_4);

    \foreach \x in {1,2,4}
    \draw [->] (x_2) --  (l_\x);

    \foreach \x in {1,2,4}
    \draw [->] (x_4) --  (l_\x);

    \draw [red,->] (l_1) -- node[above]{$w_1$} (m_1);
    \foreach \x in {2,3,5}
    \draw [->] (l_1) --  (m_\x);

    \foreach \x in {1,2,3,5}
    \draw [->] (l_2) --  (m_\x);

    \foreach \x in {1,2,3,5}
    \draw [->] (l_4) --  (m_\x);

    \draw [red,->] (m_1) -- node[above]{$w_2$} (r_1);
    \foreach \x in {2,3,5}
    \draw [->] (m_\x) --  (r_1);

    \foreach \x in {1,2,3,5}
    \draw [->] (m_\x) --  (r_3);

    \node(y_1) at (6, -1.5){$y_1$};
    \node(y_2) at (6, -2.5){$\vdots$};
    \node(y_3) at (6, -3.5){$y_m$};

    \foreach \x in {1,3}
    \draw [->] (r_\x) -- (y_\x);
\end{tikzpicture}\label{network}
    }
    \caption{Diagrams for neuron and neural network.}
    \label{nndig}
\end{figure}
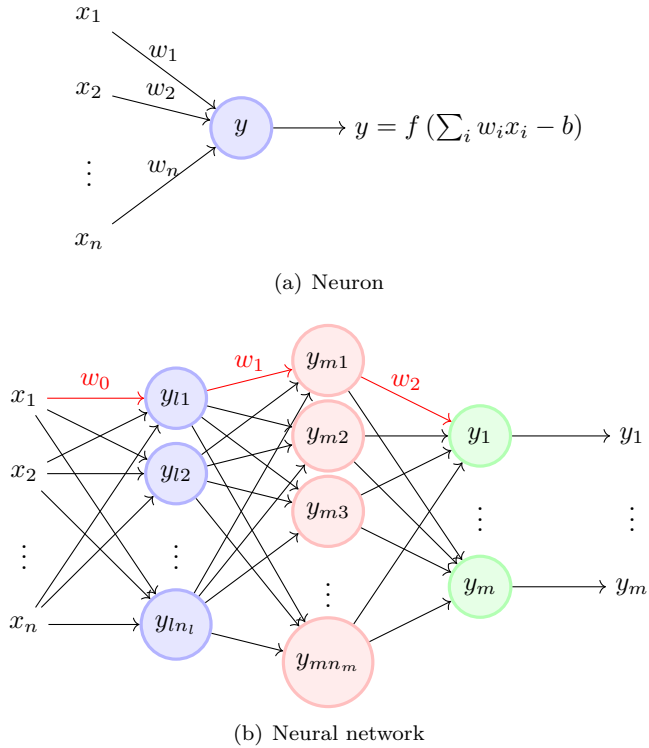

To achieve this, one can define a cost function as:
\begin{equation}\label{costclassical}
    C \equiv C(\bm{W}) := \frac 12 \sum_{i=1}^m (y_i-d_i)^2.
\end{equation}
$C=0$ implies we have finished the learning process. To find the minimum value of the cost function, one can start from some specific set of parameters and then optimize the weight vector according to optimization algorithms like gradient descent:
\begin{equation}\label{gradientdescent}
    \bm{W} \leftarrow \bm{W} - \eta \cdot \nabla C,
\end{equation}
where $\eta > 0$ is the learning rate, the gradient is $\nabla C=\{ \partial C/\partial w_j |w_j\in\bm{W} \}$. Every element in the gradient can be obtained via methods like the finite difference method:
\begin{equation}\label{finitedif}
    \frac{\partial C}{\partial w_j}=\lim_{\delta\to 0} \frac{C(w_{j\delta+})-C(w_{j\delta-})}{2\delta},
\end{equation}
where $w_{j\delta\pm}=\{ w_1,\cdots,w_j\pm\delta,\cdots\}$.

Denote the total number of weights as $M$
\footnote{The parameters number in NN and QNN may not be the same, therefore we apply different notations ($M$ and $L$).}.
If we apply Eq. \eqref{finitedif} to evaluate the gradient for every weight, we will need to execute the NN $O(M)$ times, and execute the NN once will query all $M$ weights, then the query complexity for directly evaluating the gradient scales $O(M^2)$. However, large NN execution will cost huge resources, so reducing the costs for evaluating gradients would be remarkable. We introduce the backpropagation algorithm below, which achieved this goal.

Take Fig. \ref{network} as one example, Consider the weight $w_2$, which is representative of weights corresponding to neurons in the final layer. The gradient element for this weight is:
\begin{equation}\label{gradientfinal}
    \frac{\partial C}{\partial w_2} = \frac{\partial C}{\partial y_1}
    \frac{\partial y_1}{\partial w_2}.
\end{equation}
According to Eq. \eqref{costclassical}, $\partial C/\partial y_1 = y_1-d_1$. And $\partial y_i/\partial w_2$ is the operation within one neuron, which can be easily acquired according to Eq. \eqref{neuronfunc}. 

Next, we consider evaluating the gradient concerning $w_1$, which is representative of weights in the middle layer.
\begin{equation}\label{grafinal2}
        \frac{\partial C}{\partial w_1} = \frac{\partial C}{\partial y_{m1}} \frac{\partial y_{m1}}{\partial w_1}
        = \left( \sum_i \frac{\partial C}{\partial y_i} \frac{\partial y_i}{\partial y_{m1}}  \right) \frac{\partial y_{m1}}{\partial w_1}.
\end{equation}
According to Eq. \eqref{gradientfinal},  $\partial C/\partial y_i$ is already known if all the gradients of weights corresponding to neurons in the final layers are obtained, which can be reused, and other partial derivatives are all within one neuron.
Moving back, $\partial C/\partial w_0$ can be analyzed similarly.

Therefore, when training classical NN models, one can first execute the NN and record the output ($y$) for every neuron. When evaluating gradients, weights of neurons corresponding to the final layer can be first evaluated, whose information can be reused when evaluating gradients for neurons corresponding to former layers. 
Gradient evaluation with this back-forward propagation of information is called the backpropagation algorithm, whose query complexity is $O(M)$, which establishes a reduction compared to the directly finite difference method. Using this method, we do not need to execute NNs for every weight and this makes it scalable for training NNs even with huge sizes.

\subsection{Quantum Neural Networks}\label{qnnreview}

To make it convenient for the latter analysis, we introduce the unitary coupled-cluster singles and doubles ansatz \cite{UCC} and the hardware-efficient ansatz (HEA) \cite{hea} in this section.

\subsubsection{Unitary coupled-cluster singles and doubles ansatz}\label{UCCreview}

In quantum chemistry simulations, the unitary coupled-cluster (UCC) ansatz is widely applied. It is derived from the coupled-cluster theory \cite{cc1,cc2}, which applies symmetry-conserved excitation operators on some initial states, usually the Hartree-Fock (HF) state, to expand wavefunctions in the target subspace. 

Denote the number of spin-orbitals and electrons of a given system as $n_o$ and $n_e$. And order the $n_o$ spin-orbitals from 1 to $n_o$, whose corresponding energies are in non-decreasing order. Then the HF state $ |\psi_{\operatorname{HF}}\rangle = | 1,1,\cdots,1,0,0,\cdots,0\rangle$ with exactly $n_e$ 1s and $n_o-n_e$ 0s is the state with the lowest energy when ignoring interaction energies, which is usually served as ground state approximations.

When considering the interaction energies, the ground state should be $|\psi\rangle = \sum_{ |\psi_i\rangle\in S} a_i |\psi_i\rangle$, where $a_i$ are coefficients and all states in the set $S$ satisfying the condition that the Hamming weight, i.e, the sum of all 1s is exactly $n_e$. Starting from the $|\psi_{\operatorname{HF}}\rangle$, some symmetry-conserved operations can be applied to expand the target subspace spanned by $S$. This can be realized with the fermionic creation(annihilation) operators $a_j^{\dagger}(a_j)$. For instance, the operator $a_i^{\dagger}a_{\alpha}$ can excite one electron from the $\alpha-\operatorname{th}$ spin-orbital to the $i-\operatorname{th}$ one and will result in 0 (not the vacuum state) if the $\alpha-\operatorname{th}$ orbital has no electron or the $i-\operatorname{th}$ already has one electron. Therefore, we can define it as a single-excitation operator. Double-excitation operator $a_i^{\dagger}a_j^{\dagger}a_\alpha a_\beta$ can be similarly defined.
Since considering all excitations will cost huge resources, we usually consider the single- and double-excitations, and the UCC ansatz with only the single- and double-excitation is called the UCCSD ansatz:
\begin{equation}
    |\psi_{\operatorname{UCCSD}}(\bm{\theta})\rangle = U_{\operatorname{UCCSD}}(\bm{\theta}) |\psi_{\operatorname{HF}}\rangle,
\end{equation}
where the QNN has the form:
\begin{equation}
    U_{\operatorname{UCCSD}}(\bm{\theta}) = e^{T-T^{\dagger}},
\end{equation}
where $T=T_1+T_2$ are linear combinations of excitation operators, which are expressed as:
\begin{align}
    T_1 =& \sum_{ \substack{\alpha=\{1,2,\cdots,n_e\},\\ i=\{n_e+1,\cdots,n_o\}}   } \theta_{i\alpha} a_i^{\dagger} a_{\alpha}, \\
    T_2 =& \sum_{ \substack{\alpha,\beta=\{1,2,\cdots,n_e\}, \\ i,j=\{n_e+1,\cdots,n_o\},\\ \alpha<\beta,i<j}  } \theta_{ij\alpha\beta} a_i^{\dagger} a_j^{\dagger} a_{\alpha} a_{\beta},
\end{align}
where $\bm{\theta}=\{ \theta_{i\alpha},\theta_{ij\alpha\beta} \}$ is the parameter vector. Therefore:

\begin{equation}
        T-T^{\dagger} = \sum_{\substack{\alpha=\{1,2,\cdots,n_e\},\\ i=\{n_e+1,\cdots,n_o\}} } \theta_{i\alpha} (a_i^{\dagger} a_{\alpha} - a_{\alpha}^{\dagger} a_{i}) 
        +  \sum_{\substack{\alpha,\beta=\{1,2,\cdots,n_e\}, \\ i,j=\{n_e+1,\cdots,n_o\},\\ \alpha<\beta,i<j}} \theta_{ij\alpha\beta} (a_i^{\dagger} a_j^{\dagger} a_{\alpha} a_{\beta}-a_{\beta}^{\dagger} a_{\alpha}^{\dagger} a_{j} a_{i}).
\end{equation}
To further implement the ansatz on quantum processors, fermionic-to-qubit mappings are required. We apply the Jordan-Wigner (JW) transformation \cite{jw1,jw2}. 
\begin{align}
    a_j^{\dagger} =& \frac 12 \left[\prod_{k<j}Z_k\right] (X_j-iY_j), \\
    a_j =& \frac 12 \left[\prod_{k<j} Z_k\right](X_j+iY_j).
\end{align}
After this, the HF state is mapped to $|1\rangle^{\otimes n_e}\otimes |0\rangle^{\otimes n_o-n_e}$, implying that under JW transformation, the number of qubits required is the same as the number of spin-orbitals: $n=n_o$. And the excitation operator becomes a linear combination of tensor products of Pauli operators (Pauli strings). Finally, the operation $T-T^{\dagger}$ will be a linear combination of Pauli strings. With some orders of Trotter expansion, we have:
\begin{equation}\label{ucctro}
    U_{\operatorname{UCCSD}}(\bm{\theta}) = \prod_l e^{-i\theta'_lP_l},
\end{equation}
where $\theta'$ can be obtained from $\bm{\theta}$.
For every $e^{-i\theta P}$, we can implement it on the quantum processor shown in Fig. \ref{rotpauli}.

\begin{figure}
    \centering
    \subfigure[Quantum circuit for $e^{-i\theta X_0Y_2Z_3}$]{    \begin{tikzpicture}
    \node[scale=1]{
    \begin{tikzcd}
        \lstick{$q_0$} &\gate{H}                   &\ctrl{3}\gategroup[4,steps=5,style={dashed,
        rounded corners,fill=blue!20, inner xsep=2pt},
        background]{{$e^{-i\theta Z_0Z_2Z_3}$}}  &\qw       &\qw                 &\qw       &\ctrl{3}  &\gate{H}                  &\qw \\
        \lstick{$q_1$} &\qw                        &\qw       &\qw       &\qw                 &\qw       &\qw       &\qw                       &\qw \\
        \lstick{$q_2$} &\gate{R_X(-\pi/2)} &\qw       &\ctrl{1}  &\qw                 &\ctrl{1}  &\qw       &\gate{R_X(\pi/2)} &\qw \\
        \lstick{$q_3$} &\qw                        &\targ{}   &\targ{}   &\gate{R_Z(2\theta)}  &\targ{}   &\targ{}   &\qw                       &\qw
    \end{tikzcd}
    };
\end{tikzpicture}\label{rotpauli}}
    \subfigure[Quantum circuit for HEA]{\begin{tikzpicture}
    \node[scale=1]{
    \begin{tikzcd}
    \lstick{$|0\rangle$} & \gate{R_Z}\gategroup[4,steps=7,style={dashed,
rounded corners,fill=blue!20, inner xsep=2pt},
background]{{ repeat P layers}} & \gate{R_X} & \gate{R_Z} & \ctrl{1} & \qw & \qw & \targ{}&\qw \\
    \lstick{$|0\rangle$} & \gate{R_Z} & \gate{R_X} & \gate{R_Z} & \targ{} & \ctrl{1}  & \qw & \qw & \qw\\
    \lstick{$|0\rangle$} & \gate{R_Z} & \gate{R_X} & \gate{R_Z} & \qw & \targ{} & \ctrl{1} & \qw & \qw\\
    \lstick{$|0\rangle$} & \gate{R_Z} & \gate{R_X} & \gate{R_Z} & \qw & \qw & \targ{} & \ctrl{-3} & \qw
    \end{tikzcd}
    };
  \end{tikzpicture}\label{heacir}}
    \caption{(a) Quantum circuit for $e^{-i\theta P}$ with an example of $e^{-i\theta X_0Y_2Z_3}$. The dashed part is the quantum circuit for $e^{-i\theta Z_0Z_2Z_3}$. The H gate and the $R_X$ gate are applied for basis transformation. (b) Quantum circuit for the HEA described in Eq. \eqref{heacirc}.}
\end{figure}
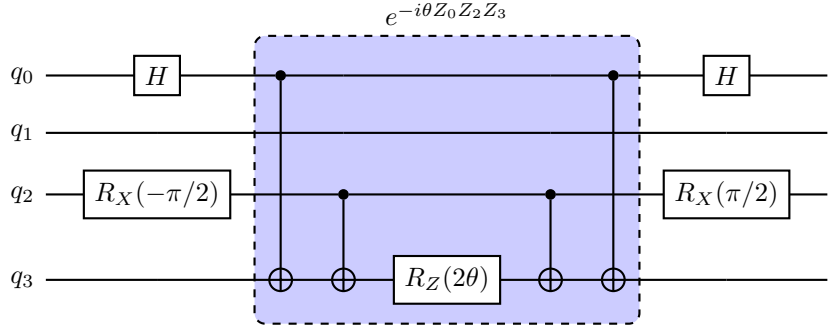
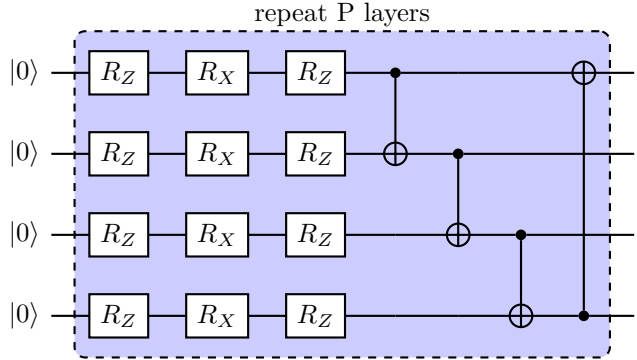

\subsubsection{Hardware-efficient ansatz}\label{heareview}

HEA is a problem-agnostic ansatz, which directly applies easy-implementable quantum gates of the quantum processor. We assume the HEA to be comprised of $P$ blocks, each of which consists of single-qubit rotation and two-qubit entangling operations:
\begin{equation}\label{heacirc}
    U_{\operatorname{HEA}}(\bm{\theta}) = \prod_{p=1}^P U_{\operatorname{entangle}} U_{\operatorname{single}}(\bm{\theta}_p),
\end{equation}
where:
\begin{align}
    U_{\operatorname{entangle}} =& \operatorname{CNOT}_{n,1} \prod_{i=1}^{n-1} \operatorname{CNOT}_{i,i+1},\\
    U_{\operatorname{single}}(\bm{\theta}_p) =& \prod_{i=1}^n R_Z(\theta_p^{i1}) R_X(\theta_p^{i2}) R_Z(\theta_p^{i3}),
\end{align}
where subscripts in CNOT gates represent the control and target qubit, respectively. The quantum circuit for the HEA described here is shown in Fig. \ref{heacir}.

It has been pointed out that HEA has remarkable expressibility \cite{EXPMPQC}. Combined with the fact that HEA is hardware-friendly, it has become the most common-applied QNN model.

\section{Gradients in variational quantum algorithms}\label{gradientindepend}

Training parameters in QNNs is the main step in executing VQAs, which is NP-hard \cite{vqanphard}. On the one hand, cost functions in VQAs are obtained via repeated measurements, and achieving sampling error $\epsilon$ will require sampling $O(1/\epsilon^2)$ times. Then about $10^6$ sampling times is required to reach the widely-applied chemical accuracy $1.6\times 10^{-3}$ Hartree
\footnote{$1 \operatorname{Hartree} = 2625.5\,\operatorname{ kJ/mol}$.}
. On the other hand, problems like barren plateaus can cause exponentially increased sampling times. Together with noise and other factors,  evaluating cost functions in VQAs would be difficult.

Note that in the training process, measuring cost function is mainly used to evaluate gradients. If we apply Eq. \eqref{finitedif} for gradient evaluation, $O(L)$ times of cost function needs to be evaluated. In Sec. \ref{backpropagation}, we introduced that the backpropagation algorithm can be used to reduce the times required for executing classical NNs, Therefore, it would be natural to ask whether such type of methods can be applied to reduce the gradient-evaluation cost when training QNNs.

First of all, the backpropagation algorithm cannot be implemented directly because a QNN is a parameterized unitary transformation that maps an initial state to the ansatz, without recording to inter-layer state, which, however, is required when performing backpropagation algorithms. As introduced in \cite{quantumback}, the backpropagation scaling for training QNNs is only possible when we have multiple copies of the ansatz.

Next, we consider whether there is some dependency between the gradient elements. If it is the case, after evaluating some gradient elements, we can apply this relation to directly compute the remaining gradient elements without running the QNN. However, we will show below that this is also unavailable.

\begin{mythm}\label{gradientind}
    For a general ansatz $U(\bm{\theta})$ with $L$ independent parameters, and the cost function defined as the expectation value under some Hamiltonian $H$, we need at least $O(L)$ times for evaluating the cost function to obtain the gradient.
\end{mythm}

The proof of this Theorem is provided below. According to this theorem, the costs for evaluating gradients in training QNNs depend on the number of parameters. This dependency heavily limits the scalability of VQAs.

In ML tasks, it is common to improve performance by increasing the number of parameters. Since there is no dependency of the gradient evaluation cost and the NN depth, such a performance-improving strategy works. However, scalability limitation makes increasing parameters not a good choice in VQAs. Since the parameter number naturally grows with the problem size or complexity, applying VQAs would be challenging.

\begin{proof}
    Suppose the PQC has the form:
    \begin{equation}\label{ansatz}
        U(\bm{\theta}) = \prod_{l=1}^L U_l(\theta_l) W_l = \prod_{l=1}^L (\cos \theta_l I - i \sin \theta_l P_l) W_l,
    \end{equation}
    where $\bm{\theta} = \{ \theta_1,\theta_2,\cdots,\theta_L \}$ is a vector of independent parameters. $P_l$ is a Hermitian operator and $W_l$ is the un-parameterized gate. Denote the initial state as $\rho_0$, then the cost function is:
    \begin{equation}\label{cost}
        C(\bm{\theta}) = \operatorname{Tr} [    U(\bm{\theta}) \rho_0 U^{\dagger} (\bm{\theta}) H   ].
    \end{equation}
    Expand Eq. \eqref{cost} according to Eq. \eqref{ansatz}, we have:
    \begin{equation}
        C(\bm{\theta}) = \operatorname{Tr} \left[   \prod_{l=1}^L (\cos \theta_l I - i \sin \theta_l P_l) W_l \rho_0 \prod_{l=L}^1 W_l^{\dagger} (\cos \theta_l I + i \sin \theta_l P_l)  H          \right].
    \end{equation}
    Observe there are 4 terms for every $\theta_l$. We view $\cos \theta_l$ and $\sin \theta_l$ as coefficients. Then the function for each term in the cost function is:
    \begin{equation}
    \begin{cases}
    \cos \theta_l \cos \theta_l, \quad & f(I,I); \\
    \cos \theta_l \sin \theta_l, \quad & f(I,iP_l); \\
    \sin \theta_l \cos \theta_l, \quad & f(-iP_l,I); \\
    \sin \theta_l \sin \theta_l, \quad & f(-iP_l,iP_l). \\
    \end{cases}
    \end{equation}
    Note that such four cases can be described by two bits $p_lq_l$ and we define the above four cases mean $p_lq_l=00,01,10,11$, respectively. Then the cost function is expressed as:
    \begin{equation}
        C = \sum_{    pq = \{ p_lq_l|p_lq_l=00,01,10,11  \}_{l=1}^L        } a_{pq} f_{pq},
    \end{equation}
    where:
    \begin{equation}
        a_{pq} = \prod_l a_{p_lq_l}, a_{p_lq_l} = \begin{cases}
            \cos^2\theta_l, & p_lq_l=00, \\
            \sin\theta_l\cos\theta_l,&p_lq_l = 01,10,\\
            \sin^2\theta_l, & p_lq_l=11.
        \end{cases}
    \end{equation}
    Denote:
    \begin{equation}
        g^l_{pq} = \frac{\partial a_{pq}}{\partial \theta_l}.
    \end{equation}
    Then the gradient is:
    \begin{equation}
        \frac{\partial C}{\partial \theta_l} = \sum_{pq} g_{pq}^l f_{pq}.
    \end{equation}
    We assume $\{f_{pq}\}$ are unknown. Computing $\partial C/\partial\theta_l$ through $\{f_{pq}\}$ requires computing almost $4^L$ times, which is impractical. 
    
    If we can obtain the full gradient by evaluating the QNN $k<O(L)$ times, then after evaluating some gradient elements we can obtain the others. Due to the unknown functions $\{f_{pq}\}$, unknown elements must be a linear combination of known gradients. If such a case exists, we consider the easiest case that we have obtained $L-1$ gradient elements, the remaining gradient can be expressed as:
    \begin{equation}
        \frac{\partial C}{\partial \theta_l} = \sum_{k\neq l} m_k  \frac{\partial C}{\partial \theta_k}.
    \end{equation}
    This means that the vectors $\{g^k_{pq}\}_{k=1}^L$ are linear dependent. Then there exists a set of numbers $\{m_i\}_{i=1}^L$ that are not all 0:
    \begin{equation}
        \sum_{l=1}^L m_l \frac{\partial C}{\partial \theta_l} = 0.
    \end{equation}
    This means:
    \begin{equation}
        \sum_{l=1}^L m_l g^l_{pq} = 0, \forall pq = \{p_lq_l\}.
    \end{equation}
    We consider the following $2^L$ elements with indices:
    \begin{equation}
        pq = \{00,11\}^L.
    \end{equation}
    And we re-order them as $w_l=p_lq_l$. Then the above equation will become:
    \begin{equation}
        \sum_{l=1}^L m_l g^l_w=0, \forall w=\{w_l\}=\{0,1\}^L.
    \end{equation}
    Define $w'=\{w_l\}_{l=2}^L$. Consider every pair of index $0,w'$ and $1,w'$, we have:
    \begin{align}
        \sum_{l=1}^L m_l g^l_{0,w'}=0, \\
        \sum_{l=1}^L m_l g^l_{1,w'}=0.
    \end{align}
    Add the two equations together:
    \begin{equation}\label{ll1}
        \sum_{l=1}^L m_l \left( g^l_{0,w'} +g^l_{1,w'} \right) =0.
    \end{equation}
    
    Observe:
    \begin{equation}
        g^l_{0,w'} + g^l_{1,w'} = \frac{\partial a_{0,w'}}{\partial \theta_l} + \frac{\partial a_{1,w'}}{\partial \theta_l} = \frac{\partial}{\partial \theta_l} (a_{0,w'}+a_{1,w'}).
    \end{equation}
    While:
    \begin{equation}
        a_{0,w'}+a_{1,w'}= \cos^2\theta_l a_{w'} + \sin^2\theta_l a_{w'}=a_{w'},
    \end{equation}
    we have:
    \begin{equation}
        g^0_{0,w'} +g^0_{1,w'} = 0.
    \end{equation}
    Then Eq. \eqref{ll1} will become:
    \begin{equation}
        \sum_{l=2}^L m_l \left( g^l_{0,w'} +g^l_{1,w'} \right) =     \sum_{l=2}^L m_l \frac{\partial a_{w'}}{\partial \theta_l} = 0.
    \end{equation}
    This is exactly the $(L-1)$-parameter case.
    Repeat this process and we will eventually have:
    \begin{equation}
        m_L \frac{\partial a_{w_L}}{\partial \theta_L} = 0, w_L=0,1.
    \end{equation}
    Since $a_{w_L=0}=\cos^2\theta_L$, $\partial a_{w_L=0}/\partial \theta_L = -\sin (2\theta_L)$. Then we have $m_L=0$ except when $\theta_l=0$. Moving back, we will obtain $m_{L-1}=0$. Finally, $m_l=0,\forall l$. This conflicts with the assumption that the vectors are linearly dependent. Then the proof is now finished.
\end{proof}

\section{Time costs for executing variational quantum algorithms}\label{costvaq}

\begin{figure}[ht]
    \centering
    \includegraphics[width=1\linewidth]{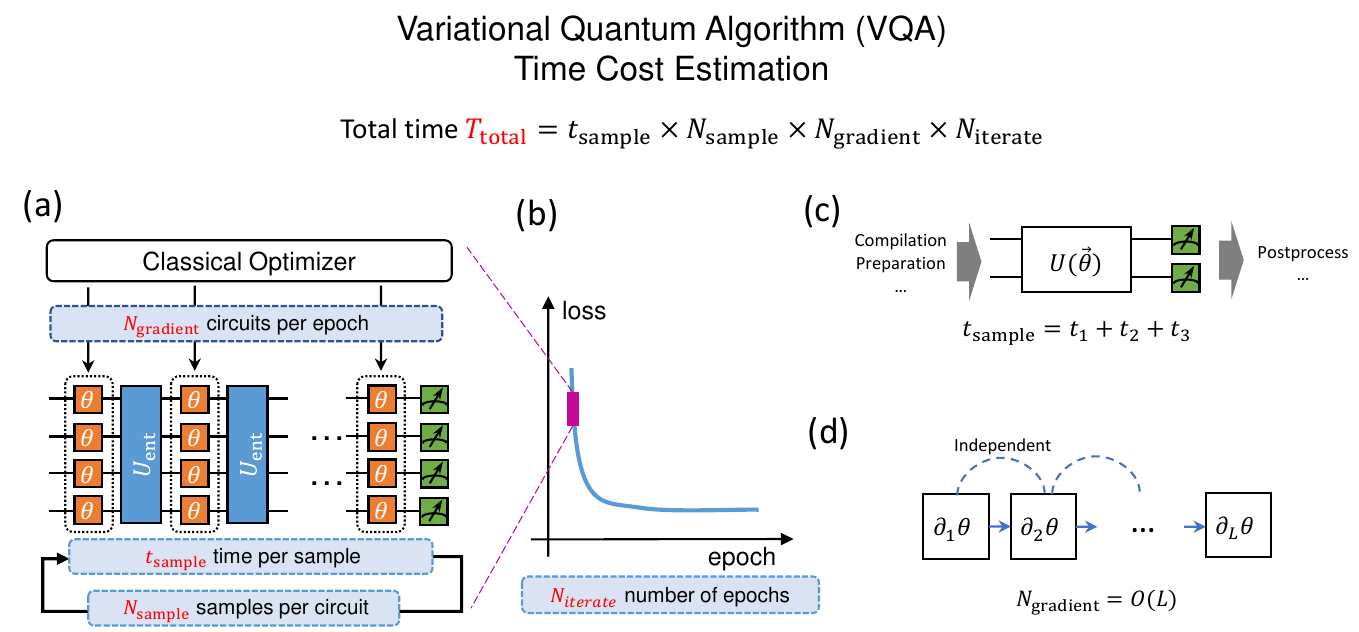}
    \caption{Sketch map for estimating the total time cost for funning VQAs.
    (a) For every step of optimization, we need to evaluate $N_{\operatorname{gradient}}$ times of cost functions to obtain the gradient. To evaluate one cost function, we need to sample the ansatz $N_{\operatorname{sample}}$ times, and the time required for each sample is $t_{\operatorname{sample}}$.
    (b) Sketch diagram for the change of loss function concerning the number of epochs. We denote the total epochs required as $N_{\operatorname{iterate}}$. The time for performing one step of optimization is shown in (a).
    (c) The process of sampling a quantum circuit, which consists of initializing the hardware, applying quantum circuits, and performing quantum measurements. Then $t_{\operatorname{sample}}$ is a sum of time for these sub-processes.
    (d) The process for evaluating gradients. Since the gradient elements are independent, we can only evaluate the partial derivatives one by one. Therefore, $N_{\operatorname{gradient}}=O(L)$.
    }
    \label{timeconsider}
\end{figure}

In this part, we estimate the time cost for executing VQAs, especially when using the UCCSD ansatz and HEA introduced in Sec. \ref{qnnreview}. Since VQA is executed by repeatedly measuring cost functions and updating parameters, the total time of running a VQA is:
\begin{equation}
    t_{\operatorname{VQA}} = t_{\operatorname{cost}} \times N_{\operatorname{cost}},
\end{equation}
where $t_{\operatorname{cost}}$ is the time needed to obtain a cost function and $N_{\operatorname{cost}}$ is the number of cost functions needed to obtain to finish the algorithm.

On the one hand, cost functions in VQAs are obtained via repeated sampling of the ansatz. Then: $t_{\operatorname{cost}} = t_{\operatorname{sample}} \times N_{\operatorname{sample}}$, where $t_{\operatorname{sample}}$ and $N_{\operatorname{sample}}$ are the time needed to sample the ansatz once and the number of samples needed to obtain a cost function, respectively. On the other hand, $N_{\operatorname{cost}}$ depends on the optimization algorithms applied. When using gradient-based algorithms, we have:
$N_{\operatorname{cost}} = N_{\operatorname{gradient}} \times N_{\operatorname{iterate}}$, where $N_{\operatorname{gradient}}$ and $N_{\operatorname{iterate}}$ are the number of cost functions needed to evaluate to obtain one gradient and the number of iteration times, respectively. Below we will analyze the above four factors. And the sketch diagram for the analysis is shown in Fig. \ref{timeconsider}.

\paragraph{$\bm{N_{\operatorname{gradient}}}$} As described in Theorem \ref{gradientind}, we can view $N_{\operatorname{gradient}}$ simply as the number of parameters in the ansatz. In the UCCSD ansatz, the number of parameters is exactly the sum of single- and double-excitation terms:
\begin{equation}\label{lucc}
    L_{\operatorname{UCCSD}} = C_{n_e}^1 C_{n_o-n_e}^1  + C_{n_e}^2 C_{n_o-n_e}^2,
\end{equation}
where 
\begin{equation}
    C_n^m = \frac{n!}{m!(n-m)!}.
\end{equation}
In HEA, parameters only appear in the single-qubit rotation operations. In each of the $P$ blocks, we apply three single-qubit gates on every qubit, then we have:
\begin{equation}
    L_{\operatorname{HEA}} = 3nP.
\end{equation}

\paragraph{$\bm{t_{\operatorname{sample}}}$} Generally, sampling a quantum circuit includes three parts: initializing the quantum hardware, running the circuit, and measuring the outcome. Then:
\begin{equation}
    t_{\operatorname{sample}} = t_{\operatorname{initial}} + t_{\operatorname{gate}} + t_{\operatorname{read}}.
\end{equation}
On current superconducting hardware, $t_{\operatorname{initial}}$
and $t_{\operatorname{read}}$ together will reach the order of 1 $\operatorname{\mu s}$ \cite{hardware1,hardware2}. The time of applying a single- and two-qubit gate are $t_{\operatorname{single}}=30\,\operatorname{ns}$ and $t_{\operatorname{double}}=60\,\operatorname{ns}$ \cite{gatetime}, respectively.
\footnote{The detailed time differs in systems but is in the same order. We will apply the averaged and experienced values.}
Then:
\begin{equation}\label{depth}
    t_{\operatorname{gate}} = l_{\operatorname{single}} \times t_{\operatorname{single}} + l_{\operatorname{double}} \times t_{\operatorname{double}},
\end{equation}
where $l$ is the single- and two-qubit gate layer depth, where two gates in the same layer indicates they can be applied at the same time. Since the time of initializing the hardware and measuring the outcome is approximate to applying $10^2$ quantum gates, then we will ignore this cost and only take the circuit running time as $t_{\operatorname{sample}}$. The following theorems provide the value of $t_{\operatorname{gate}}$ for the UCCSD ansatz and HEA.

\begin{mythm}\label{uccsddepth}
    For a many-body system with $n_o$ spin-orbitals and $n_e$ electrons, the gate layer depth for the UCCSD ansatz under the first-order Trotter expansion is:
    \begin{align}
        l_{\operatorname{single}} =& 6  C_{n_e}^1 C_{n_o-n_e}^1  +24 C_{n_e}^2 C_{n_o-n_e}^2 , \\
        l_{\operatorname{double}} =&2  n_oC_{n_e}^1 C_{n_o-n_e}^1 + \frac 83 (2n_o+1) C_{n_e}^2 C_{n_o-n_e}^2.
    \end{align}
\end{mythm}

\begin{proof}
    As introduced in Sec. \ref{UCCreview}, implementing the UCCSD ansatz on the quantum hardware requires transforming the ansatz into the form of Eq. \eqref{ucctro}. According to Fig. \ref{rotpauli}, for a $k$-local Pauli operator, which means that the operator acts non-trivially on $k$ qubits, the single-qubit and two-qubit depth of implementing $e^{-i\theta P}$ is 3 and $2k-2$, respectively. Therefore, to determine the gate layer depth with the first-order Trotter expansion, we just need to determine the number of operators $e^{-i\theta P}$ in Eq. \eqref{ucctro} and the locality for each operator $P$.

    Consider the single-excitation term, for every pair of $i>\alpha$, the single-excitation term $a_i^{\dagger} a_{\alpha} - a_{\alpha}^{\dagger} a_{i}$ is mapped with the JW transformation as:
    \begin{equation}\label{uccsmap}
        \begin{aligned}
            a_i^{\dagger}a_{\alpha} - a_{\alpha}^{\dagger}a_i = & \left[ \prod_{k<i} Z_k \right] (X_i - i Y_i) \left[ \prod_{k<\alpha} Z_k \right] (X_\alpha + i Y_\alpha)  \\
            & - \left[ \prod_{k<i} Z_k \right] (X_i + i Y_i) \left[ \prod_{k<\alpha} Z_k \right] (X_\alpha - i Y_\alpha)  \\
            =& Z_{\alpha} (X_\alpha + i Y_\alpha) \left[ \prod_{\alpha<k<i} Z_k  \right]   (X_i-i Y_i) \\
            &- Z_{\alpha} (X_\alpha - i Y_\alpha) \left[ \prod_{\alpha<k<i} Z_k  \right]   (X_i+i Y_i) \\
            =& 2 i X_\alpha \left[ \prod_{\alpha<k<i} Z_k  \right]  Y_i  - 2 i Y_\alpha \left[ \prod_{\alpha<k<i} Z_k  \right]  X_i.
        \end{aligned}
    \end{equation}
    After mapping, $a_i^{\dagger} a_{\alpha} - a_{\alpha}^{\dagger} a_{i}$ is mapped to a sum of 2 Pauli strings, each of which is $(i-\alpha+1)$-local. Similar to Eq. \eqref{uccsmap}, for every group of $i>j>\alpha>\beta$, the double-excitation term $a_i^{\dagger} a_j^{\dagger} a_{\alpha} a_{\beta}-a_{\beta}^{\dagger} a_{\alpha}^{\dagger} a_{j} a_{i}$ is mapped to a sum of 8 Pauli strings, each of which is $(i-\beta+1)$-local.

    Now we are going to determine the circuit depth. Since every $e^{-i\theta P}$ will cause 3 single-qubit circuit depth, and according to Eq. \eqref{lucc}, the number of single-excitation and double-excitation terms are $ C_{n_e}^1 C_{n_o-n_e}^1$ and $ C_{n_e}^2 C_{n_o-n_e}^2$, respectively. Then:
    \begin{equation}
        \begin{aligned}
            l_{\operatorname{single}} &= C_{n_e}^1 C_{n_o-n_e}^1 \times 2 \times 3 +  C_{n_e}^2 C_{n_o-n_e}^2 \times 8 \times 3 \\
            &=6 C_{n_e}^1 C_{n_o-n_e}^1 + 24 C_{n_e}^2 C_{n_o-n_e}^2.
        \end{aligned}    
    \end{equation}

    The case for the two-qubit depth is more complex. For every pair of $i,\alpha$, there are 2 Pauli strings for each single-excitation term, the two-qubit circuit depth for each of which is $2(i-\alpha+1)-2=2(i-\alpha)$. Therefore, the two-qubit gate layer depth with the single-excitation term is:
    \begin{equation}\label{ls}
        \begin{aligned}
            \sum_{i=n_e+1}^{n_e+(n_o-n_e)} \left(  \sum_{\alpha=1}^{n_e} 4(i-\alpha )  \right) &= \sum_{i=n_e+1}^{n_e+(n_o-n_e)} \left( 4in_e - \frac{n_e(n_e+1)}{2} \times 4\right) \\
            &= 4n_e \frac{   (n_e+1+n_o) (n_o-n_e)    }{2} -2 n_e(n_e+1)(n_o-n_e) \\
            &=2 n_on_e (n_o-n_e) \\
            &=2 n_o C_{n_e}^1 C_{n_o-n_e}^1.
        \end{aligned}
    \end{equation}

    For every group of $i,j,\alpha,\beta$, the double-excitation operator will result in 8 Pauli strings, each of which is $(i-\beta+1)$-local. And different choices of $j,\alpha$ will not affect the locality. Then the two-qubit gate depth caused by the double-excitation term is:
    \begin{equation}\label{ld}
        \sum_{i=n_e+1}^{n_e+(n_o-n_e)} \left(  \sum_{\beta=1}^{n_e} (i-\beta)(n_e-\beta)(i-n_e-1)  \right) \times 8 = \frac 83 (2n_o+1) C_{n_e}^2 C_{n_o-n_e}^2 .
    \end{equation}
    Adding Eq. \eqref{ls} and \eqref{ld}, we obtain the overall two-qubit layer depth. And the theorem is now finished.
\end{proof}

\begin{mythm}
    For the HEA described above with $P$ blocks, we have:
    \begin{align}
        l_{\operatorname{single}} &= 3P, \\
        l_{\operatorname{double}} &= nP.
    \end{align}
\end{mythm}

\paragraph{$\bm{N_{\operatorname{sample}}}$} Cost functions in VQAs are obtained via repeated sampling, where reaching the sampling error $\epsilon$ requires sampling the circuit $O(1/\epsilon^2)$ times. then $N_{\operatorname{sample}}$ is determined by the sampling accuracy required.

Generally, the sampling error should be within the accuracy required for solving the problem. However, to perform parameter optimization, sampling accuracy should also be related to the scaling of the gradient. Suppose we are applying the parameter-shift rule \cite{psr} to evaluate the gradient as:
\begin{equation}
    \partial_jC = \frac 12 \left(  C_+ -C_-   \right),
\end{equation}
with $C_{\pm} = C(\theta_j\pm \pi/2)$ and $\partial_jC = \partial C/\partial \theta_j$.

Denote the sampling error as $\epsilon$ and the sampled gradient as $\overline{\partial_jC}$. The worst case is (Suppose $\epsilon > 0$):
\begin{equation}
    \begin{aligned}
        \overline{\partial_jC} &= \frac 12 \left( [C_+ - \epsilon] - [C_-+\epsilon]   \right) \\
        &= \partial_jC - \epsilon.
    \end{aligned}
\end{equation}
To update parameters in the correct direction, we need:
\begin{equation}
    \partial_jC/\overline{\partial_jC} = \frac{\partial_jC}{\partial_jC-\epsilon} > 0.
\end{equation}
Then sampling accuracy is dependent on the scaling of the gradient.

While the magnitude of the gradient could be affected by the barren plateaus, exponential sampling times would be required, which is not workable in practice. We will analyze the time cost with a set of several given sampling times. In real tasks, we can apply methods to reduce the sampling times, address the barren plateaus phenomenon and reduce measurement costs.

\paragraph{$\bm{N_{\operatorname{iterate}}}$} Generally, $N_{\operatorname{iterate}}$ is not pre-known and differs between problems. Even for the same problem, different initial parameters and the choice of optimization algorithms will make $N_{\operatorname{iterate}}$ different. In gradient descent algorithms, both the learning rate and the gradient scaling will affect the iteration times. Moreover, while the scaling of the gradient can be affected by barren plateaus or local minimum points, optimization will take more steps. Therefore, we will treat $N_{\operatorname{iterate}}$ similar to $N_{\operatorname{sample}}$, where we will provide the time cost for a set of given $N_{\operatorname{iterate}}$. And we combine these two factors as:
\begin{equation}
    N_{\operatorname{si}} = N_{\operatorname{sample}} \times N_{\operatorname{iterate}},
\end{equation}

\paragraph{$\bm{t_{\operatorname{VQA}}}$}
Now we provide the value of $t_{\operatorname{VQA}}$ for both UCCSD ansatz and HEA. In general, 
\begin{equation}
    \begin{aligned}
        t_{\operatorname{VQA}} &=     t_{\operatorname{sample}} \times     N_{\operatorname{sample}} \times     N_{\operatorname{gradient}} \times     N_{\operatorname{iterate}} \\
        &= N_{\operatorname{si}} \times (        t_{\operatorname{single}} \times     l_{\operatorname{single}} +     t_{\operatorname{double}} \times     l_{\operatorname{double}} ) \times L \\
        &= 3\times 10^{-8} \times N_{\operatorname{si}} \times (l_{\operatorname{single}}+2l_{\operatorname{double}} )\times L.
    \end{aligned}
\end{equation}
Based on the former analysis, when considering the above ansatzes, we have:
\begin{equation}
    \begin{aligned}
        t_{\operatorname{VQA-UCCSD}} &=  10^{-8} \times N_{\operatorname{si}} \times \left(    C_{n_e}^1 C_{n_o-n_e}^1 + C_{n_e}^2 C_{n_o-n_e}^2  \right) \\
        &\times  \left[  (12n_o+18)     C_{n_e}^1 C_{n_o-n_e}^1  + (16n_o+88)C_{n_e}^2 C_{n_o-n_e}^2    \right],
    \end{aligned}
\end{equation}
and
\begin{equation}
    \begin{aligned}
        t_{\operatorname{VQA-HEA}} &= 9\times 10^{-8} \times N_{\operatorname{si}} \times (2n^2+3n)P^2 .
    \end{aligned}
\end{equation}

We can see that for a fixed $N_{\operatorname{si}}$, the total time establishes a polynomial growth.

\section{Total time cost}\label{costresult}

Based on the analysis in Sec. \ref{costvaq}, we now provide the detailed time cost for running VQAs. We will estimate the time cost under realistic assumptions of an ideal quantum processor. That is, we only take into account circuit running time and the sampling process for obtaining cost functions, and other factors including hardware noise, connectivity between physical qubits, the time for initializing the hardware and reading out the outcomes, as well as limitations for VQAs like reachability and trainability, are all ignored. The goal of ignoring these factors is to show the ``best'' time-scaling performance of VQAs.

As a representative application scenario, we consider applying VQAs to solve the ground states of different-sized molecular systems and label the systems according to their spin-orbital numbers $n_o$, which is also the number of qubits required: $n$. The number of electrons is set to be $n_e=n_o/2$.

Since $N_{\operatorname{sample}}$ and $N_{\operatorname{iterate}}$ are not pre-determined, we will provide the time cost concerning the value of the two factors, which are listed as:
\begin{align}
    N_{\operatorname{sample}}  &\in\{ 10^4,10^5,10^6,10^7,10^8 \},\\
    N_{\operatorname{iterate}} &\in\{ 10^2,10^3,10^4 \}.\label{iterate}
\end{align}
Combine them as one factor: Therefore, $N_{\operatorname{si}}$ ranges from $10^6$ to $10^{12}$.

Given $n_o$ and $n_e$, the structure of UCCSD ansatz is determined. However, the block depth $P$ needed is generally hard to be determined. Therefore, we will consider the following two cases: $P=n$ and $P=n^2$.

In Fig. \ref{tceucc} and \ref{tcehea}, we plot the time cost with different values of $N_{\operatorname{si}}$ for both UCCSD ansatz and HEA. The 1-year and 1000-year time are given as benchmarks.

\begin{figure}[htbp]
    \centering
    \includegraphics[width=0.6\linewidth]{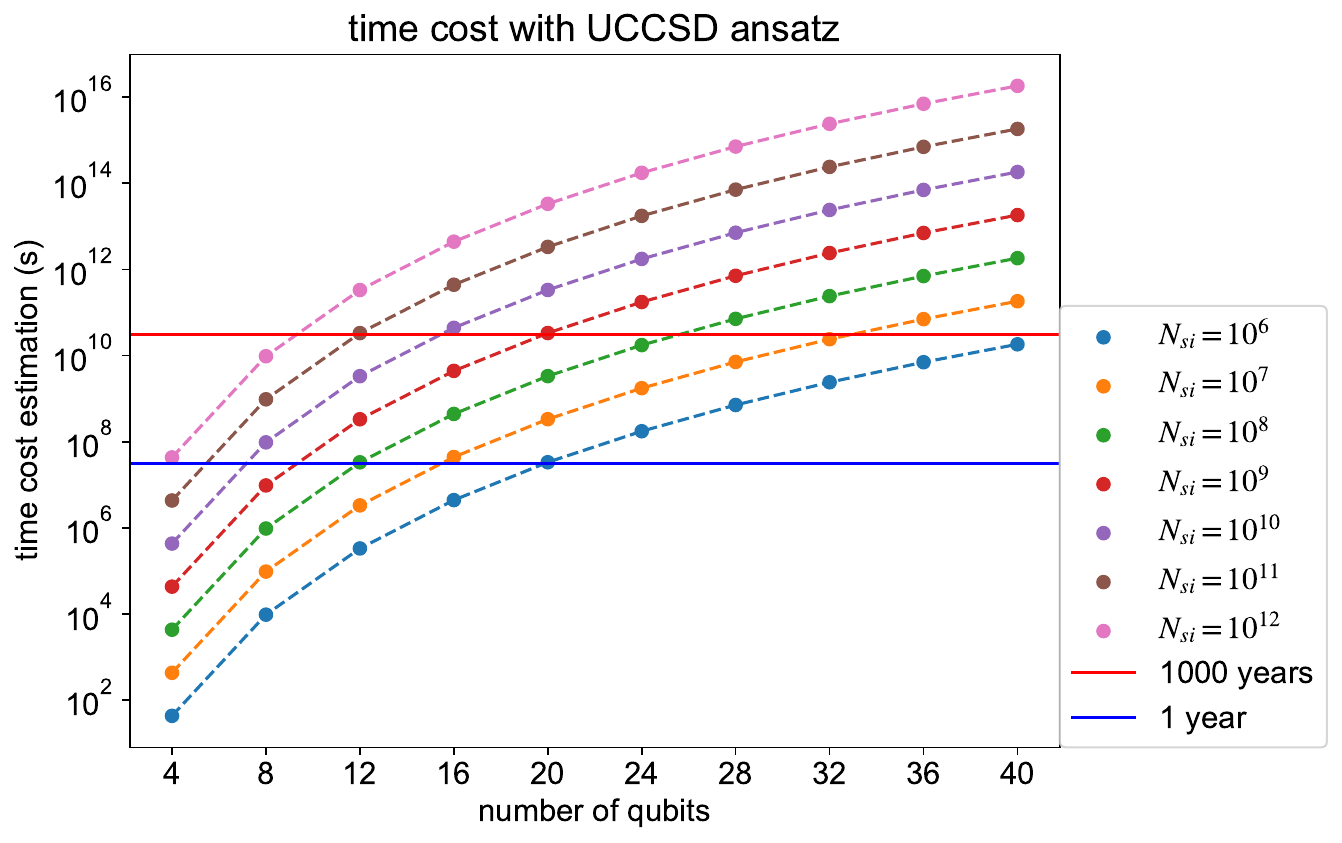}
    \caption{Time cost estimation for running a VQA with the UCCSD ansatz.}
    \label{tceucc}
\end{figure}

\begin{figure}[htbp]
    \centering
    \includegraphics[width=1\linewidth]{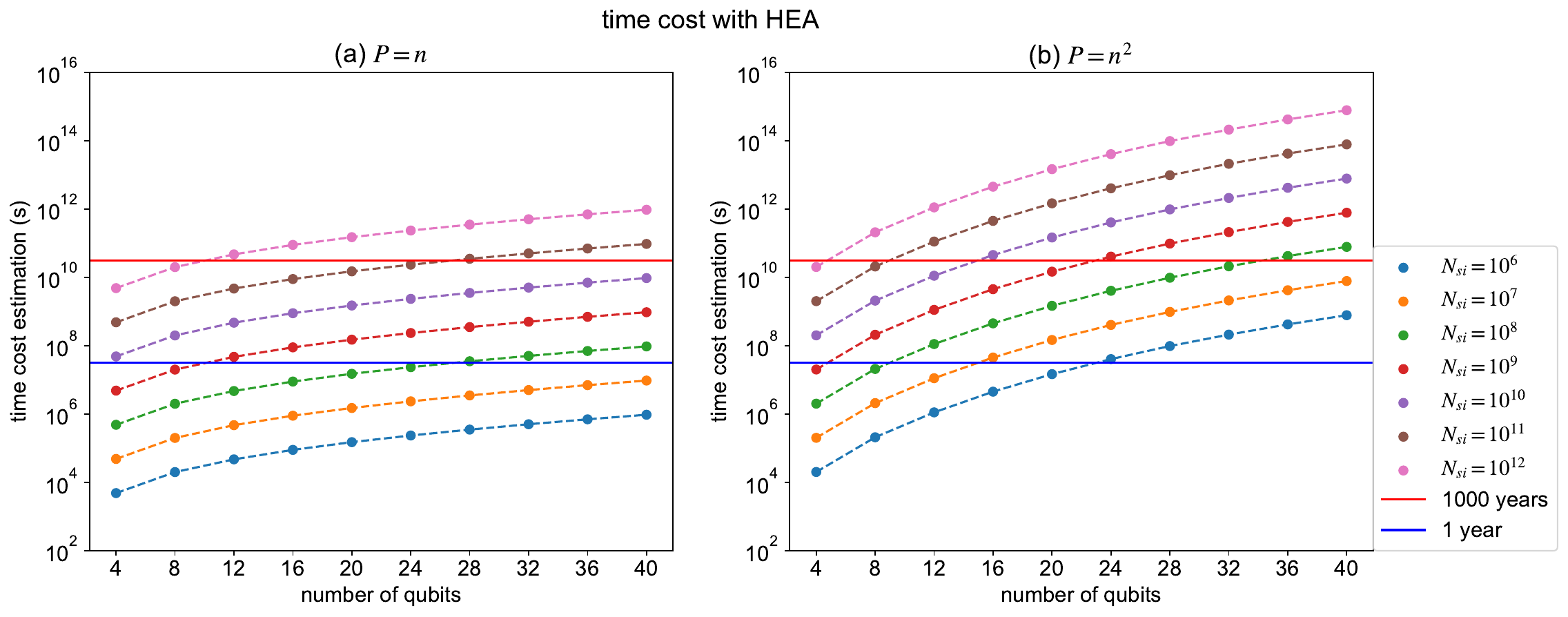}
    \caption{Time cost estimation for running a VQA with the HEA ansatz. (a) and (b) represent $P=n$ and $P=n^2$, respectively. $P$ is the number of blocks in HEA.}
    \label{tcehea}
\end{figure}

From the figures, it is clear that for a fixed value of $N_{\operatorname{si}}$, the total time cost for running VQAs establishes a polynomial growth with the number of qubits. Compared to the exponential scaling with classical simulation, VQAs seem to perform better.

However, in terms of real-time scaling, it is not the case. Even at a scaling of about 20 qubits, VQAs easily reached the 1-year time. In quantum chemistry tasks, to achieve chemical accuracy, sampling times is at least $10^6$ times. Then the total time cost corresponding to $N_{\operatorname{si}}=10^6$ can be viewed as the time for performing one step of parameter optimization, which comes at the level of 1 year. Since this is already the time on an ideal quantum computer, the real-time cost will be larger than this result. 

\section{VQAs v.s. classical simulations}\label{qvsc}

Since the term ``quantum advantage'' is a topic compared to classical simulations, it is insufficient to only provide the time cost for using VQAs. In this part, we also consider the time cost of simulating VQAs using classical simulation of quantum circuits.

As quantum processors are unavailable for common research, classical simulation of quantum circuits is widely applied. The major difference between quantum simulation and classical simulation of quantum circuits is the time of quantum gates does not change with the number of qubits, but it is not the case with classical simulation. A quantum operation $U_{\bm{x}}$ with $\bm{x}$ the list of qubits that the operation acts on, is indeed $U_{\bm{x}}\otimes I_{\bar{\bm{x}}}$, where $\bar{\bm{x}}=\{ k|k\notin \bm{x} \}$. In this case, the time of applying a quantum gate grows exponentially with the number of qubits. 

We set the gate time of 10 qubits as $t_{10}=10^{-3}$ s and the time for $n$ qubits is $t_n=t_{10}2^{n-10}$. Sampling is not required with classical simulation. We set $N_{\operatorname{sample}}=10^6$ for quantum simulations to reach the chemical accuracy. And $N_{\operatorname{iterate}}$ is listed in Eq. \eqref{iterate}.

\begin{figure}[ht]
    \centering
    \subfigure[UCCSD ansatz v.s. classical simulation]{
        \includegraphics[width=0.6\linewidth]{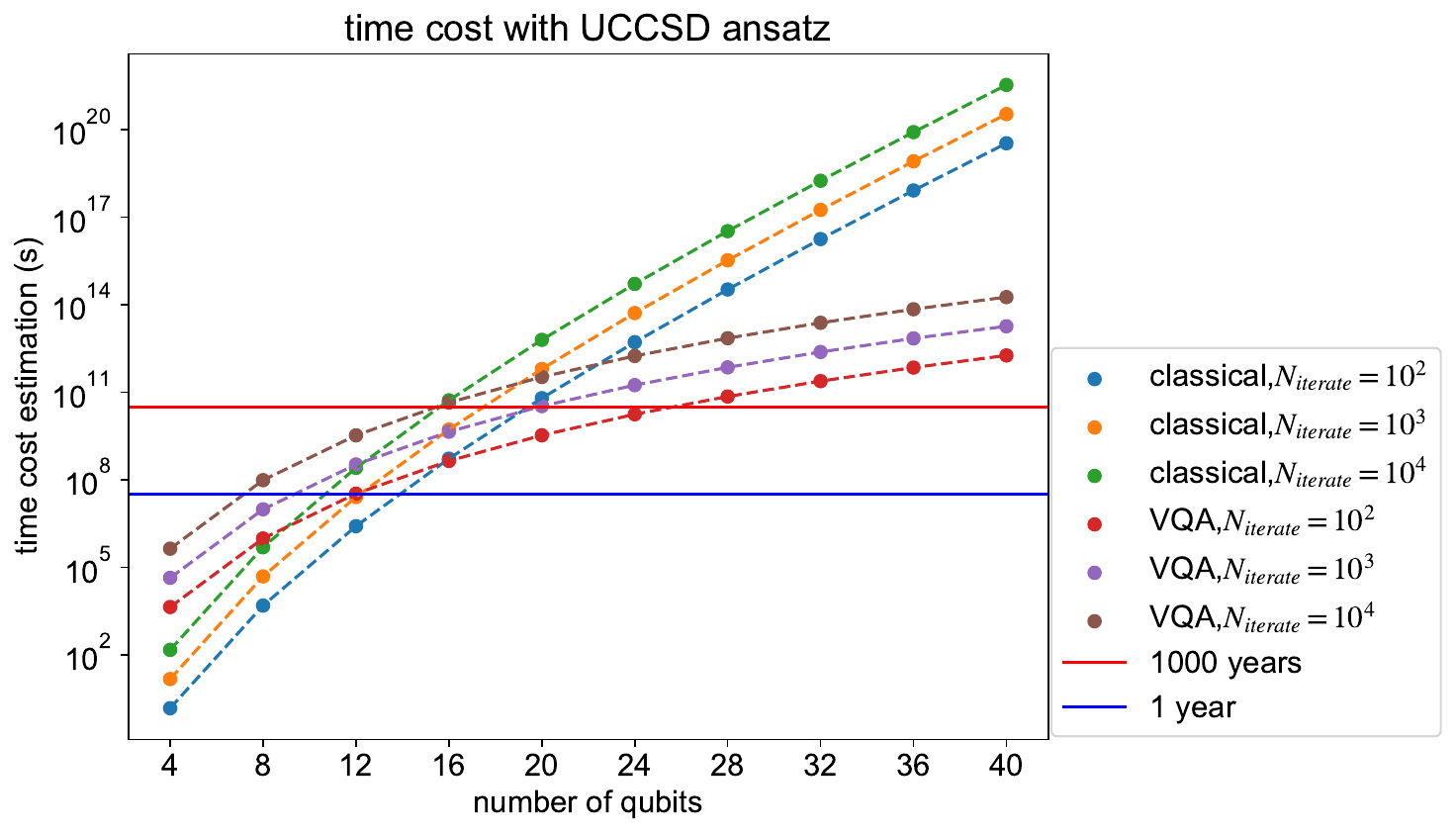}
    }
    \subfigure[HEA v.s. classical simulation]{
        \includegraphics[width=0.6\linewidth]{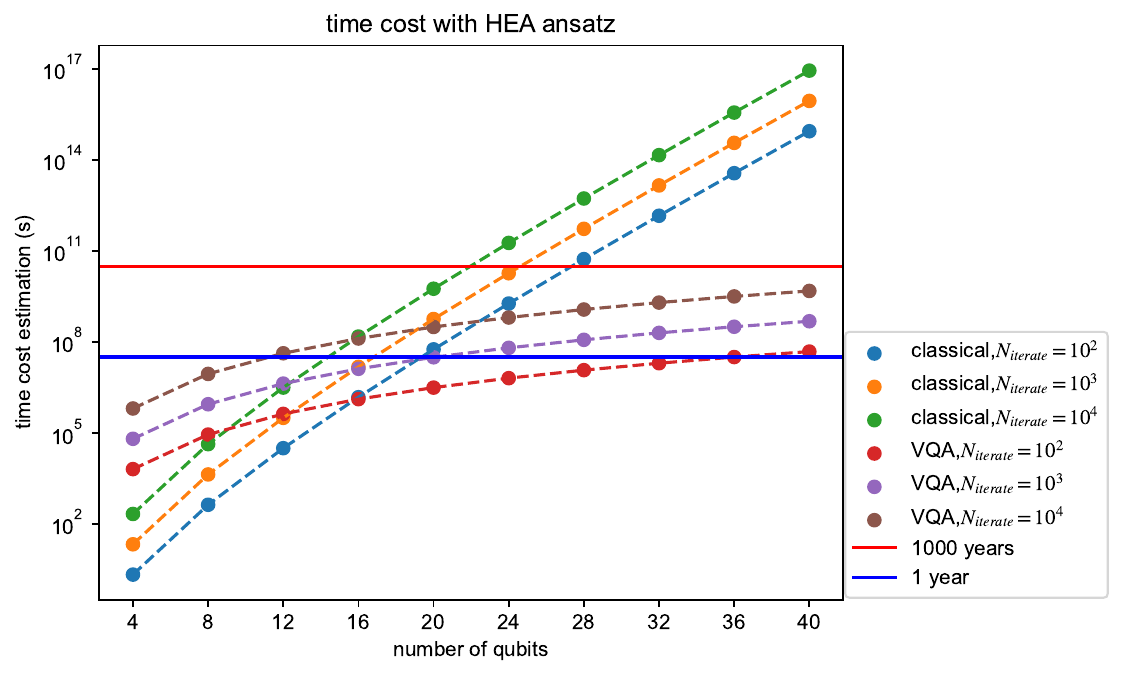}
    }
    \caption{Time cost comparison  between running VQAs and classical simulations}
    \label{quantumclassical}
\end{figure}

The time comparison between VQAs and classical simulations with both UCCSD ansatz and HEA is shown in Fig. \ref{quantumclassical}. Due to the different increasing speeds, the time curve of VQAs and classical simulations crossed, whose corresponding time is denoted as $T$, which is a function of the ansatz, iteration number, etc.
It is only possible for VQAs to outperform classical computers when the time required is larger than $T$. From the figures, this time is at the scaling of years, and it also increased with the number of parameters. 

Moreover, different from quantum processors, classical simulations can apply multi-cores, which can also provide a time reduction. For instance, in \cite{csqc}, the average gate time is 2.09 s and 1.22 s when performing a 29-qubit and 40-qubit quantum operation. While quantum simulation with multiple quantum processors is still unavailable nowadays. Therefore, quantum advantages are difficult to reach for VQAs in the acceptable time-scaling.

\section{Conclusion and outlook}\label{conclusion}

In this paper, we have investigated the time-scaling performance of VQAs and the potential for VQAs to achieve quantum advantages. We proved that methods like backpropagation cannot be directly applied when training QNNs since the inter-layer quantum states of QNNs are not recorded. And this makes the gradient-evaluation cost depend on the number of parameters in the quantum version of NN models, which limits the scalability of VQAs. Based on this result, we estimated the time cost of running VQAs in ideal cases, where realistic limitations like noise, reachability, and qubit connectivity are not considered, and we only take into account the time of performing quantum gates and errors due to finite sampling times. The result showed that even though the time established a polynomial growth, the time scaling easily reached the 1-year time wall time. Finally, we considered the time of applying classical simulations, which grows exponentially with the number of qubits. The result showed that the running time of VQAs is only shorter when the time-scaling is over $10^2$ years with the UCCSD ansatz. However, due to the realistic limitations mentioned above, whether VQAs can perform better is still not sure. At a regular time-scaling, quantum advantages may be unavailable with VQAs.

By providing such a negative comment, we do not want to deny the potential of VQAs and the NISQ algorithms. In view of VQAs, optimizations need to be made to reduce the time cost, examples like more efficient sampling strategies and more parameter-saving ansatzes. And one of our future works is to design backpropagation-type algorithms for efficiently training QNNs.

In the view of long term, introducing quantum computing into the context of machine learning, or equivalently, quantum machine learning, has remarkable potential. However, due to the different features between quantum and classical computation, directly replacing the NN model with QNN may not be the optimal way to achieve quantum advantages. Seeking a more natural way to carry out QML tasks would be meaningful.

Taking one step further, a variety of quantum algorithms is a quantum-classical hybrid: A question is solved by classical pre-processing, quantum computation, and classical post-processing. Usual algorithms replace one step of classical computation with quantum computation, but the pre-processing process to fit quantum computation is preferred.

\section*{Acknowledgement}
This work was supported by the National Natural Science Foundation of China (Grant No. 12034018), and Innovation Program for Quantum Science and Technology No. 2021ZD0302300.

\section*{Data availability}
All the data that support the findings of this study are available within this article.

\bibliography{ref.bib}
\bibliographystyle{quantum}
\end{document}